\documentclass[runningheads,svgnames]{llncs}

\usepackage{lmodern}
\usepackage{inconsolata}
\usepackage[T1]{fontenc}
\usepackage[utf8]{inputenc}
\usepackage[british]{babel}
\usepackage{csquotes}

\def\orcidID#1{\href{http://orcid.org/#1}{\protect\raisebox{-1.25pt}{\protect\includegraphics{
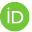%
}}}}

\usepackage[
	textsize=footnotesize,
	disable,
]{todonotes}

\usepackage[stretch=10]{microtype}
\usepackage{verbatim}
\usepackage{booktabs}
\usepackage{array}
\usepackage{tabu}

\usepackage{graphicx}
\usepackage{epstopdf}
\usepackage{amsmath}
\allowdisplaybreaks[3]
\usepackage{amssymb}
\usepackage{amsfonts}
\usepackage{mathtools}
\usepackage{siunitx}

\usepackage[shortlabels,inline]{enumitem}

\usepackage{float}
\usepackage{flafter}
\usepackage{placeins}

\usepackage[makeroom]{cancel}
	
\newcommand{\KK}[2][]{\todo[color=purple!20,#1]{KK: #2}}
\newcommand{\AD}[2][]{\todo[color=cyan!30,#1]{AD: #2}}

\newcommand{\eqv}{\mathrel{\mkern3mu\Leftrightarrow\mkern3mu}}
\newcommand{\imp}{\mathrel{\mkern3mu\Rightarrow\mkern3mu}}

\newcommand{\nimp}{\mathrel{\mkern3mu\not\Rightarrow\mkern3mu}}

\usepackage{ifthen}

\newcommand\joinable[3][]{%
	\ifthenelse{ \equal{#1}{} }%
		{#2 \mathrel{\downarrow} #3}%
		{#2 \mathrel{\downarrow_{#1}} #3}%
}
	
\newcommand\gjoinable[3][]{%
	\ifthenelse{ \equal{#1}{} }%
		{#2 \mathrel{\Downarrow} #3}%
		{#2 \mathrel{\Downarrow_{#1}} #3}%
}	
	
\newcommand\njoinable[3][]{%
	\ifthenelse{ \equal{#1}{} }%
		{#2 \mathrel{\not\downarrow} #3}%
		{#2 \mathrel{\not\downarrow_{#1}} #3}%
}

\newcommand\Vars%
	{\mathcal{V}}

\newcommand\Funcs%
	{\mathcal{F}}

\newcommand\FuncsE%
	{\mathcal{F}^e}

\newcommand\Terms%
	{\mathcal{T}}

\newcommand\AllTerms%
	{\mathcal{T}}

\newcommand\AllGTerms%
	{\mathcal{T}}
	
\newcommand\GSubs%
	{\operatorname{GSubs}}

\newcommand\GInsts%
	{\operatorname{GInsts}}
	
\newcommand\GClos%
	{\operatorname{GClos}}

\newcommand\rtrans[1]%
	{\mathrel{\overset{*}{#1}}}
	
\newcommand\trans[1]%
	{\mathrel{\overset{+}{#1}}}

\newcommand\tup[1]%
	{\langle #1 \rangle}
	
\newcommand\such%
	{\mathpunct{.}}
	
\newcommand\keyword[1]%
	{\emph{#1}}

\newcommand\beq
	{\mathbin{\approx}}
	
\newcommand\bne
	{\mathbin{\not\approx}}
	
\newcommand\bdoteq
	{\mathbin{\dot\approx}}
	
\newcommand\bto
	{\mathbin{\to}}
		
\newcommand\bmapsto
	{\mathbin{\mapsto}}

\usepackage{suffix}
\newcommand\clos[2]%
	{{{#1}\cdot{#2}}}
\WithSuffix\newcommand\clos*[2]%
	{{{(#1)}\cdot{#2}}}
	
\newcommand\clord[1]%
	{#1_\text{cl}}
	
\newcommand\litord[1]%
	{#1_l}
	
\newcommand\ssucc%
	{\mathrel{\succ\mkern-9mu\succ}}

\newcommand\sprec%
	{\mathrel{\succ\mkern-9mu\prec}}

\newcommand\mquote[1]%
	{\ensuremath{\text{`$#1\mkern-1mu$'}}}
	
\newcommand\mathwide[1]%
	{\mathrel{{\mkern3mu}{#1}{\mkern3mu}}}

\DeclareMathOperator\Exists%
	{\exists}
	
\DeclareMathOperator\Forall%
	{\forall}

\DeclareRobustCommand\infers%
	{\mathrel{\mathrel{|}\joinrel\mkern-.5mu\mathrel{-}}}

\DeclareMathOperator\mgu%
	{mgu}

\newcommand\selectedlit[1]%
	{\underline{#1}}
	
\newcommand\compl[1]%
	{\overline{#1}}

\newcommand\supterm%
	{\rhd}	

\newcommand\suptermeq%
	{\unrhd}	

\newcommand\subterm%
	{\lhd}	

\newcommand\subtermeq%
	{\unlhd}

\newcommand\relcomma%
	{\mathrel{,}}
	
\newcommand\incomparable%
	{\mathrel{\bowtie}}

\newcommand\moregen%
	{\mathrel{\sqsubset}}

\newcommand\lessgen%
	{\mathrel{\sqsupset}}

\newcommand\moregeneq%
	{\mathrel{\sqsubseteq}}

\newcommand\lessgeneq%
	{\mathrel{\sqsupseteq}}

\newcommand\id%
	{\mathit{id}}

\newcommand\mdoubleplus%
	{\mathbin{+\mkern-7mu+}}

\newcommand\concat%
	{\mdoubleplus}

\usepackage{prftree}
\usepackage{physics}

\usepackage{centernot}

\renewcommand\keyword[1]%
	{\emph{#1}}
	
\usepackage[symbol*,perpage]{footmisc}

\spnewtheorem{theorem}{Theorem}{\bfseries}{}
\spnewtheorem{lemma}{Lemma}{\bfseries}{}
\spnewtheorem{thcase}{Case}{\bfseries}{}
\spnewtheorem{thsubcase}{Subcase}[thcase]{\bfseries}{}
\spnewtheorem{definition}{Definition}{\bfseries}{}

\usepackage{soul}
\setuldepth{aa}

\newcommand\tabularbox[1]%
	{\begin{tabular}{@{}l@{}} #1 \end{tabular}}
	
\DeclareMathOperator{\acsubterms}{subterms}
\DeclareMathOperator{\acrewrite}{sort}
\DeclareMathOperator{\acnormalise}{norm}
\DeclareMathOperator{\acsort}{csort}

\definecolor{linkblue}{HTML}{0066cc}
\usepackage[
	unicode,
	colorlinks=true,
	linkcolor=MidnightBlue,
	citecolor=Green,
	urlcolor=linkblue,
]{hyperref}

\usepackage[capitalize]{cleveref}

\begin{document}

\title{AC simplifications and closure redundancies \\ in the superposition calculus} 

\author{André Duarte~\orcidID{0000-0002-5228-213X} \and Konstantin Korovin~\orcidID{0000-0002-0740-621X}}
\authorrunning{André Duarte and Konstantin Korovin}

\institute{The University of Manchester, Manchester, United Kingdom \\ 
	\{\href{mailto:andre.duarte@manchester.ac.uk}{\texttt{andre.duarte}},\href{mailto:konstantin.korovin@manchester.ac.uk}{\texttt{konstantin.korovin}}\}\texttt{@manchester.ac.uk}
}

\maketitle

\begin{abstract}
	Reasoning in the presence of associativity and commutativity (AC) is well known to be challenging due to prolific nature of these axioms. Specialised treatment of AC axioms is mainly supported by provers for unit equality which are based on Knuth-Bendix completion. The main ingredient for dealing with AC in these provers are ground joinability criteria adapted for AC.
	%
	%
	In this paper we extend AC joinability 
	from the context of unit equalities and Knuth-Bendix completion to the superposition calculus and full first-order logic. Our approach is based on an extension of the Bachmair-Ganzinger model construction and a new redundancy criterion which covers ground joinability. 
	%
	A by-product of our approach is a new criterion for applicability of demodulation which we call encompassment demodulation. This criterion is useful in any superposition theorem prover, independently of AC theories, and we demonstrate that it enables demodulation in many more cases, compared to the standard criterion.	\KK{shorten or remove: encompassment might be seen as known}

\end{abstract}

\section{Introduction}\label{sec:intro}
Associativity and commutativity (AC) axioms occur in many applications but efficient reasoning with them remain one of the major challenges in  first-order theorem proving  due to prolific nature of these axioms. 
\todo{There is a proper system description of Twee in this year's CADE, changed the citation.}
Despite a number of theoretical advances 
specialised treatment of AC axioms is mainly supported by provers for unit equality such as Waldmeister~\cite{DBLP:journals/aicom/LochnerH02}, Twee~\cite{twee-systemdesc} and MaedMax~\cite{DBLP:conf/cade/WinklerM18}. These provers are based on Knuth-Bendix completion, and the main ingredient for dealing with AC in these provers are ground joinability criteria adapted for AC~\cite{DBLP:conf/cade/MartinN90,acjoinability}. Completeness proofs for ground joinability, known so far, are restricted to unit equalities, which limits applicability of these techniques.
These proofs are based on proof transformations 
for unit rewriting which are not easily adaptable to the full first-order logic and also lack general redundancy criteria.

In this paper we extend ground AC joinability criteria  from the context of Knuth-Bendix completion to the superposition calculus for full first-order logic. Our approach is based on an extension of the Bachmair-Ganzinger model construction~\cite{DBLP:journals/logcom/BachmairG94} and a new redundancy criterion called closure redundancy. Closure redundancy allows for fine grained redundancy elimination which we show also covers ground AC joinability. We also introduced a new simplification called AC normalisation and showed that AC normalisation preserves completeness of the superposition calculus. Superposition calculus with the standard notion of redundancy can generate infinitely many non-redundant conclusions from AC axioms alone. Using our generalised notion of redundancy we can show that all of these inferences are redundant in the presence of a single extension axiom. 

Using these results, superposition theorem provers for full first-order logic such as Vampire~\cite{vampire}, E~\cite{eprover}, SPASS~\cite{spass}, Zipperposition~\cite{DBLP:conf/cade/VukmirovicBBCNT21} and iProver~\cite{iprover-systemdesc} can incorporate AC simplifications without compromising completeness.
	
A by-product of our approach is a new criterion for applicability of demodulation which we call encompassment demodulation. Demodulation is one of the main simplification rules in the superposition-based reasoning and is a key ingredient in efficient first-order theorem provers.
Our new demodulation criterion is useful independently of AC theories, and we demonstrate that it enables demodulation in many more cases, compared to the standard demodulation. 

The main contributions of this paper include: 
\begin{enumerate}
	\item New redundancy criteria for the superposition calculus called closure redundancy.
	\item Completeness proof of the superposition calculus with the closure redundancy.
	\item Proof of admissibility of AC joinability and AC normalisation simplifications for the superposition calculus. 
	\item Encompassment demodulation and its admissibility for the superposition calculus.
\end{enumerate}

\todo{General advice is more motivating examples.}
In \Cref{sec:prelim} we discuss preliminary notions, introduce closure orderings and prove properties of these orderings
In \Cref{sec:model} we introduce closure redundancy and prove the key theorem stating completeness of the superposition calculus with closure redundancy. 
In \Cref{sec:redundancies}  we use  closure redundancy to show that  encompassment demodulation,
AC joinability and AC normalisation are admissible simplifications. 
In \Cref{sec:experiment} we show some experimental results and conclude in  \Cref{sec:discussion}.

\section{Preliminaries}\label{sec:prelim}

We consider a signature consisting of a finite set of function symbols and the equality predicate as the only predicate symbol. We fix a countably infinite set of variables. 
First-order \keyword{terms} are defined in the usual manner. 
Terms without variables are called \keyword{ground terms}. 
A \keyword{literal} is an unordered pair of terms with either positive or negative polarity, 
written $s \beq t$ and $s \bne t$ respectively
(we write $s \bdoteq t$ to mean either of the former two). 
A \keyword{clause} is a multiset of literals. 
Collectively terms, literals, and clauses will be called \keyword{expressions}. 

A \keyword{substitution} is a mapping from variables to terms which is the identity for all but a finitely many variables. 
If $e$ is an expression, we denote application of a substitution $\sigma$ by $e\sigma$, 
replacing  
all variables with their image in $\sigma$. 
%
%
Let $\GSubs(e) = \{ \sigma \mid e\sigma  \text{ is ground} \}$ be the set of \keyword{ground substitutions} for $e$. 
 \KK{ground to grounding} 
Overloading this notation for sets we write $\GSubs(E) = \{ \sigma \mid \forall e \in E.~ \text{$e\sigma$ is ground} \}$. 
Finally, 
we write e.g.\ $\GSubs(e_1,e_2)$ instead of $\GSubs(\{e_1,e_2\})$. 

An injective substitution $\theta$ with codomain being the set of variables is a \keyword{renaming}. 
Substitutions which are not renamings are called \keyword{proper}. 

A substitution $\theta$ is \keyword{more general} than $\sigma$ if $\theta\rho = \sigma$ for some 
proper substitution $\rho$. 
If $s$ and $t$ can be \keyword{unified}, that is, if there exists $\sigma$ such that $s\sigma = t\sigma$, 
then there also exists the \keyword{most general unifier}, written $\mgu(s,t)$. 
%
%
A term $s$ is said to be \keyword{more general} than $t$ 
if there exists a substitution $\theta$ that makes $s\theta = t$ but there is no substitution $\sigma$ such that $t\sigma = s$. 
We may also say that $t$ is a \keyword{proper instance} of $s$. 
Two terms $s$ and $t$ are said to be \keyword{equal modulo renaming} 
if there exists a renaming $\theta$ such that $s\theta = t$. 
The relations ``less general than'', ``equal modulo renaming'', and their union 
are represented respectively by the symbols $\mquote{\lessgen}$, $\mquote{\equiv}$, and $\mquote\lessgeneq$. \KK{are ``more general than'' signs reversed ?}

A more refined notion of instance is that of \keyword{closure} \cite{bachmair_basic_1995}. 
\todo{Added reference}
Closures are pairs $\clos t\sigma$ that are said to \keyword{represent} the term $t\sigma$ while retaining information about the original term and its instantiation. 
Closures where $t\sigma$ is ground are said to be \keyword{ground closures}. 
Let $\GClos(t) = \{ \clos t\sigma \mid \text{$t\sigma$ is ground} \}$ be the set of ground closures of $t$. 
Analogously to term closures, we define closures for other expressions such as literals and clauses, as a pair of an expression and a substitution. 
Overloading the notation for sets, if $N$ is a set of clauses then $\GClos(N) = \bigcup_{C\in N} \GClos(C)$. 

We write $s[t]$ if $t$ is a \keyword{subterm} of $s$. 
If also $s \ne t$, then it is a \keyword{strict subterm}. 
We denote these relations by $s \suptermeq t$ and $s \supterm t$ respectively. 
We write $s[t\mapsto t']_p$ to denote the term obtained from $s$ by replacing $t$ at the position $p$ by $t'$.
We omit the position when it clear from the context or irrelevant. 

\todo{We clarify that we assume replacement always means “simultaneous replacement of all occurrences”; usually it is at a position do we really replace all subterms ?}

A relation $\mquote{\to}$ over the set of terms is a \keyword{rewrite relation} if (i) $l \to r \imp l\sigma \to r\sigma$ and (ii) $l \to r \imp s[l] \to s[l\mapsto r]$. 
The members of a rewrite relation are called \keyword{rewrite rules}. 
The \keyword{reflexive-transitive closure} of a relation is the smallest reflexive-transitive relation which contains it. It is denoted by $\mquote{\rtrans\to}$. Two terms are \keyword{joinable} ($\joinable{s}{t}$) if $s \rtrans\to u \rtrans\gets t$. 

 If a rewrite relation is also a strict ordering (transitive, irreflexive), then it is a \keyword{rewrite ordering}. 
A \keyword{reduction ordering} is a rewrite ordering which is well-founded.
In this paper we consider reduction orderings which are total on ground terms, 
such orderings are also \keyword{simplification orderings} 
i.e., satisfy  $s \supterm t \imp s \succ t$.


For an ordering $\mquote{\succ}$ over a set $X$, its \keyword{multiset extension} $\mquote{\ssucc}$ over multisets of $X$ is given by: 
$A \ssucc B$ iff $\forall x \in B .~ B(x)>A(x) ~ \exists y \in A .~ y \succ x \land A(y)>B(y)$, 
where $A(x)$ is the number of occurrences of element $x$ in multiset $A$. 
\todo{Is it necessary to define what “multiset” means here?}
It is well known that the mutltiset extension of a well-founded (total) order is also a well-founded (respectively, total) order~\cite{DBLP:journals/cacm/DershowitzM79}. 



\subsection*{Orderings on closures}\label{sec:orderings}

In the following, let $\mquote{\succ_t}$ be a reduction ordering which is total on ground terms. 
Examples of such orderings include KBO or LPO \cite{termrewriting}. 
\todo{Number definitions?}
It is extended to an ordering on literals via $L \succ_l L'$ iff $M_l(L) \ssucc_t M_l(L')$, 
where $M_l(s\beq t) = \{s,t\}$ and $M_l(s \bne t) = \{s,s,t,t\}$. 
It is further extended to an ordering on clauses via $C \succ_c D$ iff $C \ssucc_l D$. 

We extend this ordering to an ordering on ground closures. 
The idea is to ``break ties'', whenever two closures represent the same term, to make more general closures smaller in the ordering than more specific ones. 
The definitions follow. 
\todo{It is also suggested that we properly define how to make $\mquote{\succ_{tc}}$ total, rather than simply saying that it is extended in an arbitrary, unspecified way.}
\begin{align}
&\clos s\sigma \succ_{tc} \clos t\rho & &\text{iff} &
	&\tabularbox{
	\text{either $s\sigma \succ_t t\rho$} \\
	\text{or else $s\sigma = t\rho$ and $s \lessgen t$. }
	} \label{eq:closord:1} \\
%
%
\intertext{%
This is a well-founded ordering, since $\mquote{\succ_t}$ and $\mquote{\lessgen}$ are also well-founded. 
However it is only a partial order even on ground closures 
(e.g.,\ $\clos{f(x,b)}{(x\bmapsto a)} \incomparable \clos{f(a,y)}{(y\bmapsto b)}$), 
but it is well-known that any partial well-founded order can be extended to a total well-founded order (see e.g.\ \cite{wellfoundedext}). 
Therefore we will assume that $\mquote{\succ_{tc}}$ is extended to a total well-founded order on ground closures. 
%
Then let 
$M_{lc}(\clos*{s\beq t}\theta) = \{ \clos s\theta , \clos t\theta \}$ and 
$M_{lc}(\clos*{s\bne t}\theta) = \{ \clos s\theta , \clos {s\theta}\id , \clos t\theta , \clos {t\theta}\id \}$ in}
&\clos L\sigma \succ_{lc} \clos {L'}\rho & &\text{iff} & 
	&M_{lc}(\clos L\sigma) \ssucc_{tc} M_{lc}(\clos {L'}\rho) \,, \\
%
\intertext{and let 
$M_{cc}(\clos C\sigma) = \{ \clos L\sigma\}$ if $C$ is a unit clause $\{L\}$, and 
$M_{cc}(\clos C\sigma) = \{ \clos {L\sigma}\id \mid L \in C \}$ otherwise,\ in}
&\clos C\sigma \succ_{cc} \clos D\rho & &\text{iff} & 
	&M_{cc}(\clos C\sigma) \ssucc_{lc} M_{cc}(\clos D\rho) \,.
\end{align}

\todo{The footnote is a clarification requested by one of the reviewers, if you think it's unnecessary we can take it out.}
\KK{Let's take it out, I added a sentence in the text}
Let us note that unit and non-unit clauses are treated differently in this ordering. 
Some properties that will be used throughout the paper follow. 

\begin{lemma}\label{lemma:ord:reduction-rel}
$\mquote{\succ_{tc}}$, $\mquote{\succ_{lc}}$, and $\mquote{\succ_{cc}}$ are all well-founded 
and total on ground term closures, literal closures, and clause closures, respectively. 
\begin{proof}
We have already established that $\succ_{tc}$ is well-founded by construction. 
$\mquote{\succ_{lc}}$ and $\mquote{\succ_{cc}}$ are derived from $\mquote{\succ_{tc}}$ 
by multiset extension, so they are also well-founded. 
Similarly, $\mquote{\succ_{tc}}$ is total on ground-terms on by construction, 
 and $\mquote{\succ_{lc}}$ and $\mquote{\succ_{cc}}$ are derived from $\mquote{\succ_{tc}}$ by multiset extension, so they are also total on ground literals/clauses. 
\qed\end{proof}
\end{lemma}

\begin{lemma}\label{lemma:ord:id-id} 
Assume $s$, $t$ are ground, then
\AD{$\succ_{tc}$ is only defined on ground}
$\clos s\id \succ_{tc} \clos t\id \eqv s \succ_{t} t$.
Analogously for $\mquote{\succ_{lc}}$ and $\mquote{\succ_{cc}}$. 
\end{lemma}

\begin{lemma}\label{lemma:ord:ext}
$\mquote{\succ_{tc}}$ is an extension of $\mquote{\succ_t}$, 
in that $s\sigma \succ_t t\rho \imp \clos s\sigma \succ_{tc} \clos t\rho$, 
however this is generally not the case for $\mquote{\succ_{lc}}$ and $\mquote{\succ_{cc}}$: 
$s\sigma \bdoteq t\sigma \succ_{l} u\rho \bdoteq v\rho \nimp \clos*{s\bdoteq t}\sigma \succ_{lc} \clos*{u\bdoteq v}\rho$, 
and $C\sigma \succ_{c} D\rho \nimp \clos C\sigma \succ_{cc} \clos D\rho$. 
\end{lemma}

\begin{proof}
As an example, let $a \succ_t b$ and consider literal closures 
\begin{align}
&\clos*{f(x) \beq a}{x/a} & &\clos*{f(a) \beq b}{id}
\end{align}
The literal represented by the one on the left is greater than the one represented by the one on the right, in $\mquote{\succ_{l}}$. 
However, the closure on the left is smaller than the one on the right, in $\mquote{\succ_{lc}}$. 
This is also an example for $\mquote{\succ_{cc}}$ if these are two unit clauses. 
\qed\end{proof}

\todo{The lemma I commented here I'm positive that it is not used, or at least it is not explicitly cited.}

\begin{lemma}\label{lemma:ord:2a}
$\clos{t\rho }\sigma \succeq_{tc} \clos t{\rho\sigma}$. Analogously for $\mquote{\succ_{lc}}$ and $\mquote{\succ_{cc}}$.
In particular, $\clos{t\sigma}\id \succeq_{tc} \clos t{\sigma}$ 
and analogously for $\mquote{\succ_{lc}}$ and $\mquote{\succ_{cc}}$. 
\begin{proof}
From definition and the fact that $t\rho \lessgeneq t$. \qed\end{proof}
\end{lemma}

\begin{lemma}\label{lemma:ord:id-nid}
$\clos t\sigma \succ_{tc} \clos s\id \eqv t\sigma \succ_{t} s$.%
\footnote{But not, in general, $\clos s\id \succ_{tc} \clos t\sigma \eqv s \succ_{t} t\sigma$, e.g.\ $\clos{f(a)}\id \succ_{tc} \clos{f(x)}{(x\bmapsto a)}$.}  
Analogously for $\mquote{\succ_{lc}}$ and $\mquote{\succ_{cc}}$. 

\begin{proof}
For $\clos t\sigma \succ_{tc} \clos s\id$ to hold, either $t\sigma \succ_{t} s$, or else $t\sigma = s$ but then 
$t \lessgen s$ cannot hold. 
The $\Leftarrow$ direction follows from the definition. 
\qed\end{proof}
\end{lemma}

\begin{lemma}\label{lemma:ord:rewrite-rel}
$\mquote{\succ_{tc}}$ has the following property: 
$l \succ_{t} r \imp \clos{s[l]}\theta \succ_{tc} \clos{s[l\bmapsto r]}\theta$. 
Analogously for $\mquote{\succ_{lc}}$ and $\mquote{\succ_{cc}}$. 
\begin{proof}
For $\mquote{\succ_{tc}}$: 
let $l \succ_t r$. 
By the fact that $\mquote{\succ_t}$ is a rewrite relation, we have 
$l \succ_t r \imp
{s[l]} \succ_{t} {s[l\bmapsto r]} \imp 
{s[l]}\theta \succ_{t} {s[l\bmapsto r]}\theta$.  
Then, by the definition of $\mquote{\succ_{tc}}$, 
$\clos{s[l]}\theta \succ_{tc} \clos{s[l\bmapsto r]}\theta$. 
For $\mquote{\succ_{lc}}$ and $\mquote{\succ_{cc}}$: by the above and by their definitions we have that the analogous properties also hold. 
\qed\end{proof}
\end{lemma}

Sometimes we will drop subscripts and use just $\mquote{\succ}$ when it is obvious from the context: term, literals and clauses will be compared with $\mquote{\succ_t}$, $\mquote{\succ_l}$, $\mquote{\succ_c}$ respectively, and corresponding closures with $\mquote{\succ_{tc}}$, $\mquote{\succ_{lc}}$, $\mquote{\succ_{cc}}$.

\newcommand\cct{{\clos C\theta}}
\newcommand\cdt{{\clos D\theta}}
\newcommand\ccs{{\clos C\sigma}}
\newcommand\cds{{\clos D\sigma}}

\section{Model construction}\label{sec:model}

The superposition calculus comprises the following inference rules. 
\begin{align}
&\text{Superposition} & 
&\vcenter{\prftree%
	{\selectedlit{l \beq r\vphantom{[}} \vee C}%
	{\selectedlit{s[u] \bdoteq t} \vee D}%
	{(s[u\mapsto r] \bdoteq t \vee C \vee D)\theta}%
} 
\qcomma{\tabularbox{
	where $\theta = \operatorname{mgu}(l,u)$, \\
	$l\theta \npreceq r\theta$, $s\theta \npreceq t\theta$, \\
	and $s$ not a variable,
}} \\
&\text{Eq. Resolution} & 
&\vcenter{\prftree%
	{\selectedlit{s \bne t} \vee C}%
	{C\theta}%
} 
\qcomma{\text{where $\theta = \operatorname{mgu}(s,t)$,}} \\
&\text{Eq. Factoring} & 
&\vcenter{\prftree%
	{\selectedlit{s \beq t} \vee \selectedlit{s' \beq t'} \vee C}%
	{(s \beq t \vee t \bne t' \vee C)\theta}%
} 
\qcomma{\tabularbox{
	where $\theta = \operatorname{mgu}(s,s')$, \\ 
	$s\theta \npreceq t\theta$ and $t\theta \npreceq t'\theta$,
}}
\end{align}
and the selection function (underlined) selects at least one negative, or else all maximal (wrt.\ $\mquote{\succ_{t}}$) literals in the clause. 
\todo{Selection function}

The superposition calculus is refutationally complete wrt. the standard notion of redundancy~\cite{DBLP:journals/logcom/BachmairG94,handbook-paramodulation}. In the following, we refine the standard redundancy to closure redundancy and prove completeness in this case. 

\subsubsection{Closure redundancy}

Let $\GInsts(C) = \{ C\theta \mid \text{$C\theta$ is ground} \}$. 
In the standard definition of redundancy, a clause $C$ is redundant in a set $S$ 
if all $C\theta \in \GInsts(C)$ follow from smaller ground instances in $\GInsts(S)$. 
Unfortunately, this standard notion of redundancy does not cover many simplifications 
such as AC normalisation and a large class of demodulations 
(which we discuss in \Cref{sec:redundancies}). 

By modifying the notion of ordering between ground instances, 
using $\mquote{\succ_{cc}}$ rather than $\mquote{\succ_{c}}$, 
we adapt this redundancy notion to a closure-based one, 
which allows for such simplifications. 
We then show that superposition is still complete wrt.\ these redundancy criterion. 

A clause $C$ is \keyword{closure redundant} in a set $S$ 
if all $\clos C\theta \in \GClos(C)$ follow from smaller 
ground closures in $\GClos(S)$ 
(i.e.,\ for all $\clos C\theta \in \GClos(C)$ 
there exists a set $G \subseteq \GClos(S)$ 
such that $G \models \clos C\theta$ 
and $\Forall \clos D\rho \in G \such \clos D\rho \prec_{cc} \clos C\sigma$). 

Although the definition of closure redundancy looks similar to the standard definition, 
consider the following example showing differences between them.  
\begin{example}\label{ex:closure:redund}
Consider unit clauses $S=\{f(x)\beq g(x), g(b)\beq b\}$  where $f(x) \succ g(x) \succ b$.
Then $f(b)\beq b$ is not redundant in $S$, in the standard sense, as it does not follow from any smaller (wrt.\ $\mquote{\succ_{c}}$) ground instances of clauses in $S$,  
(it does follow from instances $f(b)\beq g(b)$, $g(b)\beq b$, but the former is bigger than $f(b) \beq b$). 
However, it is closure redundant in $S$, since its only ground instance $\clos*{f(b)\beq b}\id$ follows from 
the smaller (wrt.\ $\mquote{\succ_{cc}}$) closure instances: $\clos*{f(x)\beq g(x)}{(x\bmapsto b)}$ and $\clos*{g(b)\beq b}\id$. 
In other words, the new redundancy criterion allows demodulation even when the smaller side of the equation we demodulate with is greater than the smaller side of the target equation, provided that the matching substitution is proper. As we will see in \Cref{sec:redundancies} this considerably simplifies the applicability condition on demodulation and more crucially when dealing with theories such as AC it allows to use AC axioms to normalise clauses when standard demodulation is not be applicable. 

\todo{Reworded example a bit.}
\end{example}

%
%
Likewise, we extend the standard notion of redundant {inference}. 
An inference $C_1 , \dotsc , C_n \infers D$ is \keyword{closure redundant} in a set $S$ 
if, for all $\theta \in \GSubs(C_1,\dotsc,C_n,D)$, 
the closure $\clos{D}{\theta}$ follows from closures in $\GClos(S)$ which are smaller wrt.\ $\mquote{\succ_{cc}}$ 
than the maximal element of $\{ \clos{C_1}\theta , \dotsc , \clos{C_n}\theta \}$. 
\todo{Added notion of redundant inference}



Let us establish the following connection between closure redundant inferences and closure redundant clauses. 
An inference $C_1,\dotsc,C_n \infers D$ is \keyword{reductive} if for all $\theta \in \GSubs(C_1,\dotsc,C_n,D)$ 
we have $\clos D\theta \prec_{cc} \max\{ \clos{C_1}\theta , \dotsc , \clos{C_n}\theta \}$. 

\begin{lemma}\label{lem:redund}
If the conclusion of a reductive inference is in $S$ or is closure redundant in $S$, then the inference is closure redundant in $S$. 

\begin{proof}
If $D$ is in $S$, 
then all $\clos D\theta$ are in $\GClos(S)$. 
But if the inference is reductive then $\clos D\theta \prec_{cc} \max\{ \clos{C_1}\theta , \dotsc , \clos{C_n}\theta \}$, 
so it trivially follows from a closure smaller than that maximal element: itself. 

If $D$ is redundant, 
then all $\clos D\theta$ follow from smaller closures in $\GClos(S)$. 
But if the inference is reductive then again $\clos D\theta \prec_{cc} \max\{ \clos{C_1}\theta , \dotsc , \clos{C_n}\theta \}$, 
so it also follows from closures smaller than that maximal element. 
\qed\end{proof}
\end{lemma}


A set of clauses $S$ is \keyword{saturated up to closure redundancy} if 
any inference $C_1,\dotsc,C_n \infers D$ with premises in $S$, which are all not redundant in $S$, is closure redundant in $S$. 
In the sequel, we refer to the new notion of closure redundancy as simply ``redundancy'', when it is clear form the context.

\newcommand\rmodel[1]{R_{#1}}
\newcommand\rmodelaug[1]{R^{#1}}

\begin{theorem}\label{th:main}
The superposition inference system is refutationally complete wrt.\ closure redundancy, 
that is, if a set of clauses is saturated up to closure redundancy 
and does not contain the empty clause $\bot$, then it is satisfiable. 

%
\KK{In conclusion mention dynamic completeness for future work and cite Jasmin paper}

\begin{proof}
Let $N$ be a set of clauses such that $\bot \not\in N$, 
and $G = \GClos(N)$. 
Let us assume $N$ is saturated up to closure redundancy. 
We will build a model for $G$, 
and hence for $N$, 
as follows. 
A model is represented by a convergent term rewrite system 
(we will show convergence in \Cref{lem:model:conv}), 
where a closure $\cct$ is true in a given model $R$ if 
at least one of its positive literals $\clos*{s\beq t}\theta$ has $\joinable[R] {s\theta}{t\theta}$, 
or if at least one of its negative literals $\clos*{s\bne t}\theta$ has $\njoinable[R] {s\theta}{t\theta}$. 

For each closure $\cct \in G$,  
the partial model $\rmodel\cct$ is a rewrite system defined as $\bigcup_{\cds \prec_{cc} \cct} \epsilon_\cds$. 
The total model $\rmodel\infty$ is thus $\bigcup_{\cds \in G} \epsilon_\cds$. 
For each $\cct \in G$,
the set $\epsilon_\cct$ is defined recursively over $\prec_{cc}$ as follows. If: 
\begin{equation}\label{eq:criteria}
\begin{minipage}{0.9\textwidth}
\begin{enumerate}[a.,noitemsep]
	\item $\cct$ is false in $\rmodel\cct$,
	\item ${l\theta \beq r\theta}$ strictly maximal in $C\theta$,
	\item $l\theta \succ_t r\theta$, 
	\item $\cct \setminus \{ {\clos*{l \beq r}\theta} \}$ is false in $\rmodel\cct \cup \{l\theta \bto r\theta\}$, 
	\item $l\theta$ is irreducible via $\rmodel\cct$, 
\end{enumerate}
\end{minipage}
\end{equation}
then $\epsilon_\cct = \{ l\theta \bto r\theta \}$ and the closure is called \keyword{productive}, otherwise $\epsilon_\cct = \emptyset$. 
Let also $\rmodelaug\cct$ be $\rmodel\cct \cup \epsilon_\cct$. 
\newcommand\criteriaref[1]{(\hyperref[eq:criteria]{\ref{eq:criteria}#1})}

Our goal 
is to show that $\rmodel\infty$ is a model for $G$. 
We will prove this by contradiction: if this is not the case, 
then there is a minimal (wrt.\ $\mquote{\succ_{cc}}$) closure $\cct$ such that $\rmodel\infty \not\models \cct$. 
We will show by case analysis how the existence of this closure leads to a contradiction, 
if the set is saturated up to redundancy. 
First, some lemmas. 

\begin{lemma}\label{lem:model:conv}
$\rmodel\infty$ and all $\rmodel\cct$ are convergent, i.e.\ terminating and confluent.

\begin{proof}
It is terminating since the rewrite relation is contained in $\succ_{t}$, which is well-founded. 
For confluence it is sufficient to show that left hand sides of rules in $\rmodel\infty$ are irreducible in $\rmodel\infty$.
Assume that $l\bto r$ and $l'\bto r'$ are two rules produced by closures $\cct$ and $\cds$ respectively. 
Assume $l$ is reducible by $l'\bto r'$. 
Then $l \suptermeq l'$,  
and since $\succ_{t}$ is a simplification order, then $l \succeq_{t} l'$. 
If $l \succ_{t} l'$ 
then by \criteriaref{b} and \criteriaref{c} we have 
$l \succ_t $ all terms in $D\sigma$, 
therefore all literal closures in $\clos{D\sigma}{\id}$ will be smaller than the literal closure in $\cct$ which produced $l\bto r$ 
(by \Cref{lemma:ord:id-nid}), 
therefore $\cct \succ_{cc}  \clos{D\sigma}{\id} \succeq_{cc} \cds$ (see \Cref{lemma:ord:2a}).
But then $\cct$ could not be productive due to \criteriaref{e}. 
If $l=l'$ then both rules can reduce each other, and again due to \criteriaref{e} whichever closure is larger would not be productive. 
In either case we obtain a contradiction. 
%
%
\qed \end{proof}
\end{lemma}

\begin{lemma}\label{lem:model:up}
If $\rmodelaug\cct \models \cct$, then $\rmodel\cds \models \cct$ for any $\cds \succ_{cc} \cct$, and $\rmodel\infty \models \cct$. 
\todo{Replaced $\rmodel\cct \models \cct$ by $\rmodelaug\cct \models \cct$, since this is how it's used below.}
\end{lemma}

\begin{proof}
If a positive literal $s\beq t$ of $C\theta$ is true in $\rmodelaug\cct$, then 
$\joinable[\rmodelaug\cct] st$. 
Since no rules are ever removed during the model construction, then $\joinable[\rmodel\cds] st$ and $\joinable[\rmodel\infty] st$. 

If a negative literal $\clos*{s\bne t}\theta$ of $\cct$ is true in $\rmodelaug\cct$, 
then $\njoinable[\rmodelaug\cct]{s\theta}{t\theta}$. 
Wlog.\ assume that $s\theta\succ_t t\theta$.
Consider a productive closure $\cds \succ_{cc} \cct$ that produced a rule $l\sigma\bto r\sigma$. 
Let us show that $l\sigma\bto r\sigma$ cannot reduce $s\theta \bne t\theta$. Assume otherwise. 
By \criteriaref{b}, $l\sigma \beq r\sigma$ is strictly maximal in $D\sigma$, so 
if $l\sigma\bto r\sigma$ reduces either $t\theta$ or a strict subterm of $s\theta$, 
meaning $l\sigma \prec_t s\theta$, 
then clearly $s\theta \succ_t $ all terms in $D\sigma$, 
therefore $\clos{(s\bne t)\theta}\id \succeq_{lc} \clos*{s\bne t}\theta \succ_{lc} \text{all literals in $\clos{D\sigma}{\id}$} 
\succeq_{lc} \text{respective literals in $\clos{D}{\sigma}$}$ 
(Lemmas \ref{lemma:ord:2a} and \ref{lemma:ord:id-nid}), 
which contradicts $\cds \succ_{cc} \cct$ regardless of whether any of them is unit. 
If $l\sigma = s\theta$, 
then $M_{lc}(\clos*{s\bne t}\theta) = \{ \clos s\theta , \clos t\theta , \clos{s\theta}\id , \clos{t\theta}\id \} \ssucc_{tc} \{ \clos {l\sigma}{\id}, \clos {r\sigma}{\id} \} = M_{lc}(\clos{(l\beq r)\sigma}{\id})$, 
since $s\theta = l\sigma \succ_{t} t\sigma$ 
implies $\clos{s\theta}\id = \clos{l\sigma}\id$, and $\clos{s}{\theta} \succ_{tc} \clos{r\sigma}{\id}$. 
Hence, by \Cref{lemma:ord:2a}, ${(s\bne t)\theta}\cdot\id \succeq_{lc} \clos{(s\bne t)}{\theta} \succ_{lc} \clos{(l\beq r)\sigma}{\id} \succeq_{lc} \clos{(l\beq r)}{\sigma}$, 
contradicting $\cds \succ_{cc} \cct$ 
(again regardless of either of them being a unit).
\qed\end{proof}

\begin{lemma}\label{lem:model:upneg}
If $\cct = \clos*{C' \vee l\beq r}\theta$ is productive, 
then $\rmodel\cds \not\models \clos {C'}\theta$ for any $\cds \succ_{cc} \cct$, and $\rmodel\infty \not\models \clos {C'}\theta$. 

\begin{proof}
All literals in $\clos {C'}\theta$ are false in $\rmodelaug\cct$ by \criteriaref{d}. 
For all negative literals $\clos*{s \bne t}\theta$ in $\clos{C'}\theta$, if they are false then $\joinable[\rmodelaug\cct] {s\theta}{t\theta}$. 
Since no rules are ever removed during the model construction then $\joinable[\rmodel\cds] {s\theta}{t\theta}$ and $\joinable[\rmodel\infty] {s\theta}{t\theta}$. 

For all positive literals $\clos*{s \beq t}\theta$ in $\clos{C'}\theta$, if they are false in $\rmodelaug\cct$ then $\njoinable[\rmodelaug\cct] {s\theta}{t\theta}$. 
Two cases arise. 
If $\cct$ is unit, then $C' = \emptyset$, so $\clos{C'}\theta$ is trivially false in any interpretation. 
If $\cct$ is nonunit, then consider any productive closure $\cds \succ_{cc} \cct$ that produces a rule $l'\sigma \bto r'\sigma$, 
by definition $\cds \succ_{cc} \clos{C\theta}\id$ 
and by \Cref{lemma:ord:id-nid} $D\sigma \succ_{c} C\theta$. 
Since $l\theta \beq r\theta$ is strictly maximal in $C\theta$ 
then $l'\sigma \succ l\theta \succ \text{any term in $C\theta$}$. 
Therefore $l'\sigma \bto r'\sigma$ cannot reduce $s\theta$ or $t\theta$. 
\qed\end{proof}
\end{lemma}


We are now ready to prove the main proposition 
by induction on closures 
(see \Cref{lemma:ord:reduction-rel}), 
namely that for all $\cct \in G$ we have $\rmodel\infty \models \cct$. 
We will show a stronger result: 
that for all $\cct \in G$ we have $\rmodelaug\cct \models \cct$ 
(the former result follows from the latter by \Cref{lem:model:up}). 
\todo{Reviewer 2 says that the latter assertion is not stronger than the former, but I don't understand why he says that.}
If this is not the case, then there exists a minimal counterexample $\cct \in G$ which is false in $\rmodelaug\cct$. 

Notice that, 
since by induction hypothesis all closures $\cds \in G$ such that $\cds \prec_{cc} \cct$ have $\rmodelaug\cds \models \cds$, 
then by \Cref{lem:model:up} we have $\rmodel\cct \models \cds$ 
(and $\rmodelaug\cct \models \cds$). 
%
Consider the following cases. 
%
\begin{thcase}\label{case:1} $C$ is redundant. 
\begin{proof}
By definition, $\cct$ follows from smaller closures in $G$. 
But if $\cct$ is the minimal closure 
which is false in $\rmodelaug\cct$, 
then all smaller 
$\cds$ are true in $\rmodelaug\cds$, 
which (as noted above) means 
that all smaller $\cds$ are true in $\rmodel\cct$, 
which means $\cct$ is true in $\rmodel\cct$, 
which is a contradiction. 
%
\qed\end{proof}
\end{thcase}

\begin{thcase} $C$ contains a variable $x$ such that $x\theta$ is reducible. 
	
\begin{proof}
Then $\rmodelaug\cct$ contains a rule which reduces $x\theta$ to a term $t$.
Let $\theta'$ be identical to $\theta$ except that it maps $x$ to $t$. 
Then $C\theta' \prec C\theta$, 
so $\clos{C}{\theta'} \prec \cct$ (see \Cref{lemma:ord:ext}), 
and therefore $\clos C{\theta'}$ is true in $\rmodel\cct$. 
But $\clos{C}{\theta'}$ is true in $\rmodelaug\cct$ iff $\cct$ in $\rmodelaug\cct$, 
since $\joinable[\rmodelaug\cct] {x\theta}{t}$, 
therefore $\cct$ is also true in $\rmodelaug\cct$, which is a contradiction. 
%
\qed\end{proof}
\end{thcase}


\begin{thcase}\label{case:inference}
\todo{New case, simplifies rest of proof. [Updated]}
%
%
There is reductive inference 
$C,C_1,\dotsc \infers D$ 
which is redundant, 
such that $\{C,C_1,\dotsc\} \subseteq N$, 
$\cct$ is maximal in $\{\cct\relcomma\clos{C_1}\theta \relcomma \dotsc \}$,  
and $\cdt \models \cct$. 

\begin{proof}
Then $\cdt$ is implied by closures in $G$ smaller than $\cct$. 
But since those closures are true in $\rmodelaug\cct$, then $\cdt$ is true, 
and since $\cdt$ implies $\cct$, then $\cct$ is true in $\rmodelaug\cct$, which is a contradiction. 
\qed\end{proof}

\end{thcase}

\newcommand\refcases{\autoref{case:inference}}

\begin{thcase} Neither of the previous cases apply, and $C$ contains a \emph{negative} literal which is selected in the clause, i.e., $\clos C\theta = \clos*{C' \vee s \bne t}\theta$ with $s \bne t$ selected in $C$. 
	
\begin{proof}
Then either 
$\njoinable[\rmodel\cct] {s\theta}{t\theta}$ and $\cct$ is true and we are done, or else 
$\joinable[\rmodel\cct] {s\theta}{t\theta}$. 
Wlog., let us assume 
$s\theta \succeq t\theta$.

\begin{thsubcase} $s\theta = t\theta$. 

\begin{proof}
Then $s$ and $t$ are unifiable, 
meaning that there is an equality resolution inference 
\begin{equation}\label{key}
C' \vee s \bne t \infers C'\sigma \qc{\text{with $\sigma = \mgu(s,t)$,}}
\end{equation}
with premise in $N$. 

Take the instance $\clos{C'\sigma}\rho$ of the conclusion such that $\sigma\rho = \theta$; it always exists since $\sigma=\mgu(s,t)$. 
\todo{Idempotence of mgu is necessary to match the definition of redundant inference.}
Also, since the mgu is idempotent \cite{termrewriting} 
then $\sigma\theta = \sigma\sigma\rho = \sigma\rho$, 
so $\clos{C'\sigma}\rho = \clos{C'\sigma}\theta$. 
We show that $\cct = \clos*{C' \vee s\bne t}{\sigma\rho} \succ \clos{C'\sigma}\rho = \clos{C'\sigma}\theta$. 
If $C'$ is empty, then this is trivial. 
If $C'$ has more than 1 element, then this is also trivial (see \Cref{lemma:ord:id-id}). 
If $C'$ has exactly 1 element, then let 
$C' = \{ s'\bdoteq t' \}$. 
We have $\clos*{s'\bdoteq t' \vee s\bne t}{\sigma\rho} \succ \clos{(s'\bdoteq t')\sigma}\rho$
if $\clos{(s'\bdoteq t')\sigma\rho}\id                  \succeq \clos{(s'\bdoteq t')\sigma}\rho$, 
which is true by \Cref{lemma:ord:2a}. 
Notice also that if $\clos{C'\sigma}\rho$ is true then $\clos*{C' \vee \dotsb}{\sigma\rho}$ must also be true. 

Recall that \refcases\ does not apply. 
But we have shown that this inference is reductive, 
with $C \in N$, 
$\cct$ trivially maximal in $\{\cct\}$, 
and that the instance $\clos{C'\sigma}\theta$ of the conclusion implies $\cct$. 
So for \autoref{case:inference} not to apply the inference must be non-redundant. 
Also since \autoref{case:1} doesn't apply then the premise is not redundant. 
This means that the set is not saturated, which is a contradiction. 
\qed\end{proof}
\end{thsubcase}

\begin{thsubcase} $s\theta \succ t\theta$. \label{case:eqres1}
\begin{proof}
Then (recall that $\joinable[\rmodel\cct] {s\theta}{t\theta}$) $s\theta$ must be reducible by some rule in $\rmodelaug\cct$. 
Since by \criteriaref{b} the clause cannot be productive, 
it must be reducible by some rule in $\rmodel\cct$. 
Let us say that this rule is $l\theta \bto r\theta$, 
produced by a closure $\cdt$ 
smaller than $\cct$.%
\footnote{We can use the same substitution $\theta$ on both $C$ and $D$ by simply assuming wlog.\ that they have no variables in common.} 
Therefore closure $\cdt$ must be of the form $\clos{(D' \vee l\beq r)}\theta$, 
with $l\theta \beq r\theta$ maximal in $D\theta$, 
and $\clos{D'}\theta$ false in $\rmodel\cdt$. 
Also note that $\cdt$ cannot be redundant, or else it would follow from smaller closures, 
but those closures (which are smaller than $\cdt$ and therefore smaller than $\cct$) 
would be true, so $\cdt$ would be also true in $\rmodel\cdt$, so by \criteriaref{a} it would not be productive. 

Then $l\theta = u\theta$ for some subterm $u$ of $s$,  
meaning $l$ is unifiable with $u$, 
meaning there exists a superposition inference
\begin{equation}\label{key}
D' \vee l \beq r \relcomma 
C' \vee s[u] \bne t \infers
(D' \vee C' \vee s[u\bmapsto r] \bne t)\sigma
\qc{\text{$\sigma=\mgu(l,u)$,}}
\end{equation}

Similar to what we did before, consider the instance $(D' \vee C' \vee s[u\bmapsto r] \bne t)\sigma \cdot \rho$ with $\sigma\rho = \theta$.%
\footnote{And again note that the $\mgu$ $\sigma$ is idempotent so $(D' \vee C' \vee s[u\bmapsto r] \bne t)\sigma \cdot \rho = (D' \vee C' \vee s[u\bmapsto r] \bne t)\sigma \cdot \theta$.}
%
%
We wish to show that 
this instance of the conclusion is smaller than $\cct$ (an instance of the second premise), 
that is that
\begin{equation}
\clos{(C' \vee s \bne t)}{\sigma\rho} \;\succ\; \clos{(D' \vee C' \vee s[u\bmapsto r] \bne t)\sigma}\rho \,.
\end{equation}
Several cases arise: 
\begin{itemize}[wide,nosep,label=\textbullet]

\item $C' \ne \emptyset$. 
Then both premise and conclusion are non-unit, so comparing them means comparing 
$C'\theta \vee s\theta\bne t\theta$ and $D'\theta \vee C'\theta \vee s\theta[u\theta \bmapsto r\theta] \bne t\theta$ 
(\Cref{lemma:ord:id-id}), 
or after removing common elements, comparing 
$s\theta\bne t\theta$ and $D'\theta \vee s\theta[u\theta \bmapsto r\theta] \bne t\theta$. 
This is true since (i) $l\theta \succ r\theta 
\imp s\theta[l\theta] \succ s\theta[l\theta\bmapsto r\theta] 
\imp s\theta \bne t\theta \succ s\theta[l\theta\bmapsto r\theta] \bne t\theta$, 
and (ii) 
$s\theta \succeq l\theta \succ r\theta$ and 
$l\theta \beq r\theta$ is greater than all literals in $D'\theta$, 
so $s\theta \bne t\theta$ is greater than all literals in $D'\theta$. 


\item $C' = \emptyset$ and $D' \ne \emptyset$. 
Then we need 
$\clos{(s \bne t)}{\sigma\rho} \succ \clos{(D' \vee s[u\bmapsto r] \bne t)\sigma\rho}\id$. 
By \Cref{lemma:ord:id-nid}, this is true only if 
$s\theta \bne t\theta \succ D'\theta \vee s\theta[u\theta \bmapsto r\theta] \bne t\theta$. 
To see that this is true we must also notice that, 
since $\cdt \prec \cct$, then 
(again by \Cref{lemma:ord:id-nid}) 
$D'\theta \vee l\theta \beq r\theta \prec s\theta \bne t\theta$ 
must also hold, 
so $\{s\theta \bne t\theta\} \succ D'\theta$. 
Then obviously $\{s\theta \bne t\theta\} \succ \{s\theta[u\theta \bmapsto r\theta] \bne t\theta\}$. 

\item $C' = \emptyset$ and $D' = \emptyset$. 
Then simply $s\theta[u\theta] \succ s\theta[u\theta\bmapsto r\theta]$ 
means $\clos {s[u]}{\sigma\rho} \succ \clos {s[u\bmapsto r]\sigma}{\rho}$, 
which since $s\sigma\rho \succ t\sigma\rho$, 
means $\clos*{s[u]\bne t}{\sigma\rho} \succ \clos{(s[u\bmapsto r]\bne t)\sigma}\rho$. 
\end{itemize}

%
%

In all these cases this instance of the conclusion is always smaller than the instance $\cct$ of the second premise. 
%
%
%
Note also that $\cct$ is maximal in $\{\cct\relcomma\cdt\}$. 
Also, since $\clos {D'}\theta$ is false in $\rmodel\cct$ (by \Cref{lem:model:upneg}) 
and $(s[u\bmapsto r] \bne t)\cdot\theta$ is false in $\rmodel\cct$ 
(since $\clos*{s \bne t}\theta$ is in the false closure $\cct$, $\joinable[\rmodel\cct]{u\theta}{r\theta}$, and the rewrite system is confluent), 
then in order for that instance of the conclusion to be true in $\rmodel\cct$ it must be the case that $\clos{C'\sigma}\rho$ is true in $\rmodel\cct$. 
But if the latter is true then $\cct = \clos*{C' \vee \dotsm}{\sigma\rho}$ is true, in $\rmodel\cct$. 
In other words that instance of the conclusion implies $\cct$. 
Therefore again, since \autoref{case:1} and \autoref{case:inference} don't apply, we conclude that the inference is non-redundant with non-redundant premises, so the set is not saturated, which is a contradiction. 
\qed\end{proof}
\end{thsubcase} 

\noindent This proves all subcases. 
\qed\end{proof}
\end{thcase}

\begin{thcase} Neither of the previous cases apply, so all selected literals in $C$ are positive, 
i.e., $\clos C\theta = \clos*{C' \vee s \beq t}\theta$ 
with $s \beq t$ selected in $C$. 
\todo{Here relate maximality of $s\beq t$ in $C$ and $s\theta \beq t\theta$ in $C\theta$}

\begin{proof}
Then, since if the selection function doesn't select a negative literal then it must select all maximal ones, 
wlog.\ one of the selected literals $s\beq t$ must have $s\theta \beq t\theta$ is maximal in $C\theta$. 
Then if either 
$\clos{C'}\theta$ is true in $\rmodel\cct$, 
or $\epsilon_\cct = \{ s\theta \bto t\theta \}$, 
or $s\theta = t\theta$, 
then $\cct$ is true in $\rmodelaug\cct$ and we are done. 
Otherwise, 
$\epsilon_\cct = \emptyset$, 
$\clos{C'}\theta$ is false in $\rmodel\cct$, 
and wlog.\ $s\theta \succ t\theta$. 
If $s\beq t$ is maximal in $C$ then $s\theta \beq t\theta$ is maximal in $C\theta$. 

\begin{thsubcase} $s\theta \beq t\theta$ maximal but not strictly maximal in $C\theta$. 

\begin{proof}
If this is the case, then there is at least one other maximal positive literal in the clause. 
Let $\cct = \clos{(C'' \vee s\beq t \vee s'\beq t')}\theta$, 
where $s\theta = {s'}\theta$ and $t\theta = {t'}\theta$. 
Therefore $s$ and $s'$ are unifiable and there is an equality factoring inference:
\begin{equation}\label{key}
C'' \vee s\beq t \vee s'\beq t' \infers
(C'' \vee s\beq t \vee t\bne t')\sigma
\qc{\text{with $\sigma = \mgu(s,s')$,}}
\end{equation}
with $\sigma = \mgu(s,s')$. 
%
%
Take the instance of the conclusion $\clos{(C'' \vee s\beq t \vee t\bne t')\sigma}\rho$ with $\sigma\rho = \theta$. 
This is smaller than $\cct$ 
(since $s'\theta \beq t'\theta \succ t\theta \bne t'\theta$, and \Cref{lemma:ord:id-id} applies). 
Since $t\theta = t'\theta$ and $\clos{C''\sigma}\rho$ is false in $\rmodel\cct$, 
this instance of the conclusion is true in $\rmodel\cct$ iff $\clos*{s\sigma \beq t\sigma}\rho$ is true in $\rmodel\cct$. 
But if the latter is true in $\rmodel\cct$ then $\clos*{s\beq t \vee \dotsb}{\sigma\rho}$ also is. 
Therefore that instance of the conclusion implies $\cct$. 
As such, and since again Cases \ref{case:1} and \ref{case:inference} do not apply, we have a contradiction. 
\qed\end{proof}
\end{thsubcase}

\begin{thsubcase} $s\theta \beq t\theta$ strictly maximal in $C\theta$, and $s\theta$ reducible (in $\rmodel\cct$). 

\begin{proof}
This is similar to \autoref{case:eqres1}. 
If $s\theta$ is reducible, say by a rule $l\theta \bto r\theta$, 
then (since $\epsilon_\cct = \emptyset$) this is produced by some closure $\cdt$ smaller than $\cct$, 
with $\cdt = \clos*{D' \vee l \beq r}\theta$, 
with the $l\theta \beq r\theta$ maximal in $D\theta$, and with $\clos {D'}\theta$ false in $\rmodel\cdt$. 

Then there is a superposition inference 
\begin{equation}\label{key}
D' \vee l \beq r \relcomma 
C \vee s[u] \beq t \infers
(D' \vee C' \vee s[u\bmapsto r] \beq t)\sigma
\qc{\text{$\sigma=\mgu(l,u)$,}}
\end{equation}
Again taking the instance $(D' \vee C' \vee s[u\bmapsto r] \beq t)\sigma \cdot \rho$ with $\sigma\rho = \theta$, 
we see that it is smaller than $\cct$ (see discussion in \autoref{case:eqres1}).  
%
%
Furthermore since $\clos {D'}\theta$ and $\clos{C'}\theta$ are false in $\rmodel\cct$,  
then that instance of the conclusion is true in $\rmodel\cct$ iff $\clos{(s[u\bmapsto r] \beq t)\sigma}\rho$ is. 
But since also $\joinable[\rmodel\cct]{u\theta}{r\theta}$, 
then $\clos{(s[u\bmapsto r] \beq t)\sigma}\rho$ implies $\clos{(s[u] \beq t)\sigma}\rho$. 
Therefore that instance of the conclusion implies $\cct$. 
Again this means we have a contradiction. 
\qed\end{proof}
\end{thsubcase}

\begin{thsubcase} $s\theta \beq t\theta$ strictly maximal in $C\theta$, and $s\theta$ irreducible (in $\rmodel\cct$). 

\begin{proof}
Since $\cct$ is not productive, 
and at the same time all criteria in \eqref{eq:criteria} except \criteriaref{d} are satisfied, 
it must be that \criteriaref{d} is not, 
that is $\clos{C'}\theta$ must be true in $\rmodelaug\cct = \rmodel\cct \cup \{ s\theta \bto t\theta \}$. 
Then this must mean we can write $\clos{C'}\theta = \clos*{C'' \vee s'\beq t'}\theta$, 
where the latter literal is the one that becomes true with the addition of $\{s\theta \bto t\theta\}$, 
whereas without that rule it was false. 

But this means that $\joinable[\rmodelaug\cct] {s'\theta}{t'\theta}$ such that any rewrite proof needs at least one step where $s\theta \bto t\theta$ is used, since $s\theta$ is irreducible by $\rmodel\cct$.
Wlog.\ say $s'\theta \succ t'\theta$. 
Since: 
(i) $s\theta \beq t\theta \succ s'\theta \beq t'\theta$, 
(ii) $s\theta \succ t\theta$, and
(iii) $s'\theta \succ t'\theta$, 
then $s\theta \succeq s'\theta \succ t'\theta$, which implies $t'\theta \centernot{\suptermeq} s\theta$, which implies $s\theta \bto t\theta$ can not be used to reduce $t'\theta$. 
Then the only way it can reduce $s'\theta$ or $t'\theta$ is if $s\theta = s'\theta$. 
This means there is an equality factoring inference:
\begin{equation}\label{key}
C'' \vee s'\beq t' \vee s\beq t \infers
(C'' \vee s'\beq t' \vee t\bne t'
)\sigma
\qc{\text{with $\sigma = \mgu(s,s')$.}}
\end{equation}
Taking $\theta = \sigma\rho$, we see that the instance of the conclusion 
$\clos{(C'' \vee t\bne t' \vee s\beq t)\sigma}\rho$ 
is smaller than the instance of the 
$\clos{(C'' \vee s'\beq t' \vee s\beq t )}\sigma\rho$. 

%
But we have said that $\joinable[\rmodelaug\cct] {s'\theta}{t'\theta}$, 
where the first rewrite step had to take place by rewriting $s'\theta = s\theta \to t\theta$, 
and the rest of the rewrite proof then had to use only rules from $\rmodel\cct$. 
In other words, this means $\joinable[\rmodel\cct] {t\theta}{t'\theta}$. 
As such, the literal $\clos*{t\bne t'}\theta$ is false in $\rmodel\cct$, 
and so the conclusion is true in $\rmodel\cct$ iff rest of the closure is true in $\rmodel\cct$. 
But if the rest of the closure $\clos{(C'' \vee s'\beq t')\sigma}\rho$ then so is $\cct$, 
so that instance of the conclusion implies $\cct$. 
Once again, this leads to a contradiction since none Cases \ref{case:1} and \ref{case:inference} apply and therefore the set must not be saturated. 
\qed\end{proof}
\end{thsubcase}

\noindent This proves all the subcases and the theorem.
\qed\end{proof}
\end{thcase}

\end{proof}
\end{theorem}

\noindent\emph{Remark:} As part of this proof we have also shown that all inferences in the superposition system are reductive, so per \Cref{lem:redund} one way to make inferences redundant is simply to add the conclusion. 


\section{Redundancies}\label{sec:redundancies}

Now we will show three novel redundancy criteria whose proof is enabled by the framework we have just discussed. 
One is an extension of the demodulation rule, used in many different provers. 

\subsection*{Demodulation} \label{sec:demodulation}

Recall the ``standard'' demodulation rule 
(a struck clause means that it can be removed from the set when the conclusion is added). 
\begin{align}
& \text{Demodulation} &
& \vcenter{\prftree{l \beq r}{\cancel{C[l\theta]}}{C[l\theta \bmapsto r\theta]}} 
\qcomma{\begin{tabular}{@{}l@{}}
		where $l\theta \succ r\theta$ \\ 
		and $\{l\theta\beq r\theta \} \prec C[l\theta]$.
\end{tabular}}
\\ \intertext{We show an extension which is also a redundancy in this framework.}
& \parbox{1em}{Encompassment \\ Demodulation} &
& \vcenter{\prftree{l \beq r}{\cancel{C[l\theta]}}{C[l\theta \bmapsto r\theta]}} 
\qcomma{\begin{tabular}{@{}l@{}}
	where $l\theta \succ r\theta$, and \\ 
	either $\{l\theta\beq r\theta \} \prec C[l\theta]$  \\
	or $l\theta \lessgen l$.
\end{tabular}}
\end{align}

\begin{theorem}\label{th:encom:demod}
Encompassment demodulation is a sound and admissible simplification rule wrt.\ closure redundancy 
(a redundancy criterion is admissible if its struck premises are redundant wrt.\ the conclusion and the non-struck premises). 
\end{theorem}

\KK{Appendix references as in the extended version of this paper \cite{}...\url{https://www.cs.man.ac.uk/~korovink/my_pub}}

\begin{proof}
This is a valid redundancy if all ground closures of $C[l\theta]$ follow from some smaller ground closures of $C[l\theta \bmapsto r\theta]$ and $l\beq r$. 
That they follow is trivial. 
The fact that if $l\theta \succ r\theta$ then, for all groundings $\rho$, $\clos{C[l\theta \bmapsto r\theta]}\rho \prec \clos{C[l\theta]}\rho$, also follows from \Cref{lemma:ord:rewrite-rel}. 
It remains to show that, for all $\rho$, $\clos{\{l\beq r\}}{\theta\rho} \prec \clos{C[l\theta]}\rho$.

If there exists a literal $s \bdoteq t \in C$ such that $s \succ l\theta$, 
then (remember $l\theta \succ r\theta$) 
$l\theta \beq r\theta \prec s\bdoteq t \imp 
\{l\theta\beq r\theta\} \prec C$. 
In this case, 
for all ground closures $\clos C\rho$ we have 
the following: 
$s\rho \succ l\theta\rho \succ r\theta\rho 
\imp \clos s\rho \succ \clos l{\theta\rho} \succ \clos r{\theta\rho} 
\imp \clos*{s\bdoteq t}\rho \succ \clos*{l\beq r}{\theta\rho} 
\imp \clos C\rho \succ \clos{\{l\beq r\}}{\theta\rho}$. 

If not, 
hence
$l\theta$ only occurs at a top position in a literal in $C$, 
then all occurrences of $l\theta$ in $C$ are of the form $l\theta \bdoteq s$.  
If there exists at least one such literal where $l\theta \prec s$, 
then again $l\theta \beq r\theta \prec l\theta \bdoteq s \imp \{l\theta\beq r\theta\} \prec C$, and also 
for all $\rho$ we have 
$s\rho \succ l\theta\rho \succ r\theta\rho$ 
as above, and therefore 
$\clos C\rho \succ \clos{\{l\beq r\}}{\theta\rho}$. 

If not,  
hence $l\theta$ always occurs at the top of a maximal side of an equality, 
then still if $l\theta \lessgen l$ 
then for all $\rho$, $\clos{l\theta}\rho \succ \clos l{\theta\rho}$. 
Therefore $\clos*{l\theta \bdoteq s}\rho \succ \clos*{l\beq r}{\theta\rho}$, 
so $\clos C\rho \succ \clos{\{l\beq r\}}{\theta\rho}$ for all $\rho$. 
\AD{Unit/non-unit: Is it necessary to repeat this bit in every case? When we will have nice lemma to refer to we can skip repeating this}
This holds whether or not $C$ is unit, by \Cref{lemma:ord:2a}. 

If not, then 
(since then neither $l\theta \lessgen l$ nor obviously $l \lessgen l\theta$) 
we check if $l\theta \bdoteq s$ is negative. 
If yes, then 
$\{ l\theta, r\theta \} \prec \{l\theta , l\theta , s , s\}$, 
so again 
it is the case that $\{l\theta \beq r\theta\} \prec C$, and also that for all $\rho$, 
$\{ \clos{l\theta}\rho , \clos{l\theta\rho}\id , \clos s\rho , \clos{s\rho}\id \} \succ \{ \clos{l\theta}\rho , \clos{r\theta}\rho \}$, 
therefore 
$\clos*{l\theta \bne s}\rho \succ \clos*{l\beq r}{\theta\rho}$, 
therefore $\clos C\rho \succ \clos{\{l\beq r\}}{\theta\rho}$. 

If not, then we finally need to compare 
$s$ and $r\theta$. 
If $s \succ r\theta$, then 
$l\theta \beq r\theta \prec l\theta \beq s \imp 
\{l\theta \beq r\theta\} \prec C$, 
and also, for all $\rho$, 
$s\rho \succ r\theta\rho 
\imp \clos s\rho \succ \clos r{\theta\rho} 
\imp \clos*{l\theta \beq s}\rho \succ \clos*{l\beq r}{\theta\rho} 
\imp \clos C\rho \succ \clos{\{l\beq r\}}{\theta\rho}$. 
%

If not, then we have $l\theta \not\lessgen l$ and $\{l\theta \beq r\theta\} \nprec C$, so the simplification cannot be applied. 
\qed\end{proof}

This theorem has many practical implications. 
Demodulation is widely used in superposition theorem provers, 
and improvement this criterion provides are two-fold.   

First, it enables strictly more simplifying inferences to be performed where they previously could not.
Let us re-consider \Cref{ex:closure:redund} from \Cref{sec:model}. 
Standard demodulation is not applicable to $f(b)\beq b$ by clauses in  $S=\{f(x)\beq g(x), g(b)\beq b\}$.
However,  we can simplify it to a tautology and remove it completely using encompassment demodulation. Our experimental results (\Cref{sec:experiment}) show that  encompassment demodulation extends usual demodulation in many practical problems. 

Second, it enables a faster way to check the applicability conditions. One of the considerable overheads in the standard demodulation is to check that the equation we are simplifying with is smaller than the clause we are simplifying. 
For this, right-hand side of the oriented equation needs to be compared in the ordering with all top terms in the clause.  
In the encompassment demodulation this expensive check is avoided in many cases. 
After obtaining the matching instantiation $\theta$ of the left side of the oriented equation, if it is not a renaming (a quick check) or the matching is strictly below the top position of the term,  then we can immediately accept the inference and skip potentially expensive ordering checks. 

\subsection*{Associative-commutative joinability} \label{sec:acjoin}

Let ${AC}_f$ be
\begin{subequations} \label{eq:acall}
\begin{align}
f(x,y) &\beq f(y,x) \,, \label{eq:ac1} \\
f(x,f(y,z)) &\beq f(f(x,y),z) \,, \label{eq:ac2} \\
f(x,f(y,z)) &\beq f(y,f(x,z)) \,. \label{eq:ac3}
\end{align}
\end{subequations}
The first two axioms (\ref{eq:ac1}) and (\ref{eq:ac2}) define that $f$ is an associative-commutative (AC) symbol. The third equation (\ref{eq:ac3}) follows from those two and will be used to avoid any inferences between these axioms and more generally to justify AC joinability simplifications defined next.

We define the two following rules: 
\begin{subequations}
\begin{align}
& \text{AC joinability (pos)} &
& \vcenter{\prftree{\cancel{s \beq t \vee C}}{AC_f}{}} 
\qcomma{\tabularbox{
	where $\joinable[AC_f]st$ \\
	$s\beq t \vee C$ not in $AC_f$, 
}} \label{eq:acjoin:pos} \\
& \text{AC joinability (neg)} &
& \vcenter{\prftree{\cancel{s \bne t \vee C}}{AC_f}{C}} 
\qcomma{\tabularbox{
	where $\joinable[AC_f]st$, 
}} \label{eq:acjoin:neg}
\end{align}
\end{subequations}

\begin{theorem}\label{thm:ac:join}
AC joinability rules are sound and admissible simplification rules wrt.\ closure redundancy. 
\KK{move proof into appendix ?}
\end{theorem}
\begin{proof}
	
Let us prove rule \eqref{eq:acjoin:pos}. 
We will show how, 
if $\joinable[AC_f]st$, 
then all ground instances $\clos*{s\beq t}\theta$ 
are rewritable, via smaller instances of clauses in $AC_f$, 
to a smaller tautology or to a smaller instance of clauses in $AC_f$, 
meaning that $s\beq t$ is redundant wrt.\ closure redundancy. 
Using closure redundancy is essential, as instances of $AC_f$ axioms used in the following rewriting process can be bigger than the clause we are simplifying in the usual term ordering, but as we will see they are smaller in the closure ordering. 

For conciseness, let us denote $f(a,b)$ by $ab$ 
in the sequel. 
We will assume that the term ordering has following properties: 
if $s\succ_{t} t$ then $st \succ_{t} ts$ and $s(tu) \succ_{t} t(su)$, 
and also that $(xy)z \succ_{t} x(yz)$. 
This conditions hold for most commonly used families of orderings, such as KBO or LPO \cite{termrewriting}. 

First some definitions. 
Let $\acsubterms_f$ collect all ``consecutive'' $f$-subterms into a multiset, that is 
\begin{subequations}
\begin{align}
&\text{if $u = f(s,t)$:} & &\acsubterms_f(u) = \acsubterms_f(s) \cup \acsubterms_f(t) \,, \\
&\text{otherwise:}       & &\acsubterms_f(u) = u \,.
\end{align}
\end{subequations}
so for example $\acsubterms_f(a((bc)d)) = \{a,b,c,d\}$. 
%
Let us define $\acrewrite_f$ as follows: 
\begin{alignat}{2}
&\acrewrite_f(u) & &= u''_1(\dotsm u''_n)\qcomma{\tabularbox{%
	where $\{u_1,\dotsc,u_n\} = \acsubterms_f(u)$ \\
	and $u'_i = \acrewrite_f(u_i)$ \\
	and $\{u''_1,\dotsc\} = \{u'_1,\dotsc\}$, \\
	and $u''_1 \prec \dotsb \prec u''_n$. 
}}
\end{alignat}
such that for example if $a\prec b \prec c$ then $\acrewrite_f(\,(ba)(g(cb))\,) = a(b(g(bc)))$. 
Note that we have 
$\joinable[AC_f] st \imp \forall \theta\in\GSubs(s,t) \such~ \joinable[AC_f] {s\theta}{t\theta}$, 
and $\joinable[AC_f] st \eqv \acrewrite_f(s) = \acrewrite_f(t)$. 
Therefore we will now show how, if $\joinable[AC_f] st$, then for any ground instance $\clos*{s\beq t}\theta$, 
the closure $\clos*{s' \beq t'}\theta$, 
with $s'\theta = \acrewrite_f(s\theta) = \acrewrite_f(t\theta) = t'\theta$, 
is either an instance of $AC_f$ or a tautology, 
implied by smaller instances of clauses from $AC_f$. 
\todo{Possibly is all of this $\acrewrite_f$ stuff not necessary? Note that we are stating $\joinable[AC_f] st \eqv \acrewrite_f(s) = \acrewrite_f(t)$ without proof, so this is already some of the proof that we're skipping. So maybe we just write “we will show how we can reduce $\clos*{s\beq t}\theta$ to $\clos*{s'\beq t'}\theta$ with $s'\theta = t'\theta$”.}

\newcommand\cardinality[1]{\abs{#1}}


\newcommand\mybox[1]{\parbox[t]{0.60\textwidth}{#1}}

For the cases where $\cardinality{\acsubterms_f(s)}$ is 1, 2, or 3, 
ad-hoc proofs are required. 
%
Let $\theta$ be any grounding. 
\begin{subequations}
\begin{align}
x &\beq x & &\mybox{Tautology} \\
xy &\beq xy & &\mybox{Tautology} \\
xy &\beq yx & &\mybox{Instance of \eqref{eq:ac1}} \\
\midrule
x(yz) &\beq x(yz) & &\mybox{Tautology} \\
x(yz) &\beq y(zx) & &\mybox{%
	If $x\theta \prec z\theta$, rewrite $zx\to xz$ to get an instance of \eqref{eq:ac3}. 
	If $z\theta \prec x\theta$ and $z\theta \prec y\theta$, rewrite $y(zx)\to z(xy)$ to get an instance of \eqref{eq:acs:3}. 
	If $y\theta \prec z\theta \prec x\theta$, rewrite $x(yz)\to z(yx)$ — using smaller \eqref{eq:acs:3} — to get an instance of \eqref{eq:ac3}. 
} \label{eq:acs:1} \\
x(yz) &\beq z(xy) & &\mybox{%
	If $y\theta \prec x\theta$, rewrite $xy\to yx$ to get an instance of \eqref{eq:acs:3}. 
	If $z\theta \prec y\theta$, rewrite $yz\to zy$ to get an instance of \eqref{eq:acs:3}. 
	If $x\theta \prec y\theta \prec z\theta$, rewrite $z(xy)\to y(xz)$ — using smaller \eqref{eq:acs:3} — to get an instance of \eqref{eq:ac3}.
} \\
x(yz) &\beq y(xz) & &\mybox{Instance of \eqref{eq:ac3}} \\
x(yz) &\beq x(zy) & &\mybox{%
	If $y\theta \prec z\theta$, rewrite $zy \to yz$. 
	If $z\theta \prec y\theta$, rewrite $yz\to zy$. 
	In both cases, we reach a tautology.
	} \label{eq:acs:2} \\
x(yz) &\beq z(yx) & &\mybox{%
	If $z\theta \prec y\theta$ and $x\theta \prec y\theta$, rewrite $yz\to zy$ and $yx\to xy$ to get an instance of \eqref{eq:ac3}. 
	If $y\theta \prec z\theta$ and $y\theta \prec x\theta$, rewrite $z(yx)\to y(zx)$ to get an instance of \eqref{eq:acs:3}. 
	If $x\theta \prec y\theta \prec z\theta$, rewrite on the right: $yx\to xy$, then $z(xy)\to x(zy)$, then $zy\to yz$ to obtain a tautology. 
	If $z\theta \prec y\theta \prec x\theta$, rewrite on the left: $yz\to zy$, then $x(zy)\to z(xy)$, then $xy\to yx$ to obtain a tautology. 
} \label{eq:acs:3} \\
\midrule
(xy)z &\beq x(yz) & &\mybox{Instance of \eqref{eq:ac2}} \\
(xy)z &\beq y(zx) & &\mybox{%
	If $x\theta \prec y\theta$, rewrite $(xy)z\to x(yz)$ to get an instance of \eqref{eq:acs:1}. 
	If $y\theta \prec x\theta$, rewrite $xy\to yx$ to get $(yx)z\beq y(zx)$, rewrite $(yx)z\to y(xz)$ to get an instance of \eqref{eq:acs:2}.
	} \\
(xy)z &\beq z(xy) & &\mybox{Instance of \eqref{eq:ac1}} \\
(xy)z &\beq y(xz) & &\mybox{%
	If $x\theta \prec y\theta$, rewrite $y(xz)\to x(yz)$. 
	If $y\theta \prec x\theta$, rewrite $xy\to yx$. 
	In both cases, we reach an instance of \eqref{eq:ac2}.
	} \\
(xy)z &\beq x(zy) & &\mybox{%
	If $y\theta \prec z\theta$, rewrite $zy\to yz$ to get an instance of \eqref{eq:ac2}. 
	If $z\theta \prec y\theta$, rewrite $(xy)z \to z(xy)$ (via a proper instance of $xy\beq yx$, that is) to get an instance of \eqref{eq:ac3}.
	} \\
(xy)z &\beq z(yx) & &\mybox{%
	If $x\theta \prec y\theta$, rewrite $yx \to xy$.
	If $y\theta \prec x\theta$, rewrite $xy\to yx$. 
	In both cases, we reach an instance of \eqref{eq:ac1}.
	} 
\end{align}
\end{subequations}
Then by \Cref{lemma:ord:2a} all cases with $\cardinality{\acsubterms(s)} \le 3$ follow, 
since they will be an (equal or more specific) instance of some such case. 
\AD{Is it clear what is meant here? I mean if we prove all $\clos*{x(yz)\beq (yx)z}\theta$ follow from smaller instances then by \Cref{lemma:ord:2a} all $\clos{(x(yz)\beq (yx)z)\rho}\theta$ will also follow, and that covers all equations with $\#\acsubterms \le 3$}


For the cases with $\cardinality{\acsubterms_f(s)} \ge 4$, 
consider any ground instance $\clos*{s\beq t}\theta$. 
First, exhaustively apply the rule $(xy)z \to x(yz)$ 
on all subterms of $s \beq t$. 
Since $(xy)z \succ x(yz)$, 
$s \succeq s'$ and $t \succeq t'$, 
then (\Cref{lemma:ord:rewrite-rel}) 
$\clos*{s\beq t}\theta \succeq \clos*{s'\beq t'}\theta$. 
In order to show that $\clos*{s'\beq t'}\theta$ and $AC_f$ make $\clos*{s\beq t}\theta$ redundant, 
it remains to be shown that these rewrites were done by instances of \eqref{eq:ac2} which are also 
smaller than $\clos*{s\beq t}\theta$. 

Since $\cardinality{\acsubterms_f(s)} \ge 4$, 
then any $s$ or $t$ where we can rewrite with $(xy)z\bto x(yz)$ is in one of the following forms: 
(i) $(a_1a_2)(a_3a_4)$, 
in which case we can use an identical argument to encompassment demodulation since $(a_1a_2)(a_3a_4) \lessgen (xy)z$, 
or (ii) $a_1a_2$ with the term being rewritten being 
$a_2$ or a subterm thereof, 
in which case the rewrite is also by a smaller instance.  

After this, $s'$ and $t'$ are of the form 
$a_1(\dotsm a_n)$. 
Now, since the closure is ground, 
for every adjacent pair of terms either 
$a_i\theta \prec a_{i+1}\theta$ or 
$a_i\theta \succ a_{i+1}\theta$ or 
$a_i\theta = a_{i+1}\theta$. 
This means we can always instantiate and apply one of \eqref{eq:ac1} or \eqref{eq:ac3} 
\KK{[This is next paragraph right?] is instantiation proper ?}
and ``bubble sort'' the AC terms until they become 
$a_1'(\dotsm a_n')$ with $a_1'\theta \prec \dotsb \prec a_n'\theta$, 
where there is a bijection between 
$\{a_1,\dotsc,a_n\}$ and $\{a_1',\dotsc,a_n'\}$, 
obtaining an $a_1'\theta(\dotsm a_n'\theta) \preceq a_1\theta(\dotsm a_n\theta)$. 

Once again, these rewrites are done via smaller instances of $AC_f$, since 
we either 
rewrite with \eqref{eq:ac1} on a subterm, in the case of $a_{n-1}/a_n$, 
or with \eqref{eq:ac3} on a subterm, in the case of $a_i/a_{i+1}$ with $2 \le i \le n-1$, 
or with \eqref{eq:ac3} on a less general term, in the case of $a_1/a_2$. 

The process we have just described is done bottom-up on terms 
(meaning for instance $f(g(f(b,a)),c) \to f(g(f(a,b)),c) \to f(c,g(f(a,b)))$). 
Obviously, the rewrites on inner $f$-subterms are trivially done by smaller instances. 

This concludes the process. 
Applying this on both sides yields the closure
$\clos*{s' \beq t'}\theta$ 
with $s'\theta = \acrewrite_f(s\theta)$ and $t'\theta = \acrewrite_f(t\theta)$, 
which we have shown is $\preceq \clos*{s\beq t}\theta$ 
and follows from it by smaller closures in $\GClos(AC_f)$. 
\todo{For example here simply note that we get a closure $\clos*{s'\beq t'}\theta$ with $s'\theta = t'\theta$, and that's it, no need to invoke $\acrewrite_f$.}
This can be done for all $\theta \in \GSubs(s,t)$. 
Thus $\acrewrite_f(s,\theta) \beq \acrewrite_f(t,\theta)$, a tautology, 
makes clause $s\beq t$ redundant, 
meaning any $s\beq t \vee C$ 
is redundant in $AC_f$. 
The same process proves rule \eqref{eq:acjoin:neg}. 
\qed\end{proof}


%


\subsection*{AC normalisation}

We will now show some examples to motivate another simplification rule. 
Assume $a\prec b\prec c$. 
The demodulation rule already enables us to rewrite any occurrence of, for instance, 
$b(ca)$, or $(ac)b$ or any other such permutation, to $a(bc)$. 
However, take the term $b(xa)$. 
It cannot be simplified by demodulation. 
Yet it is easy to see that in any instance of a clause where it appears, 
it can be rewritten 
to a smaller $a(xb)$ 
via smaller instances of clauses in $AC_f$. 

Such cases motivate the following simplification rule.%
\footnote{Note we trivially assume all $AC_f$ terms are right associative, since $(xy)z\bto x(yz)$ is always oriented.}
\begin{align}
&\text{AC norm.} &
&\vcenter{\prftree{\cancel{C[t_1(\dotsm t_n)]}}{AC_f}{C[t'_1(\dotsm t'_n)]}}
\qcomma{\tabularbox{%
	where $t_1,\dotsc,t_n \succ_\text{lex} t'_1,\dotsc,t'_n$ \\
	and $\{t_1,\dotsc,t_n\} = \{t'_1,\dotsc,t'_n\}$ 
}}
\end{align}

\begin{theorem}\label{thm:ac:norm} 
AC normalisation is a sound and admissible simplification rule wrt.\ closure redundancy. 
\end{theorem}


\begin{proof} 
The conclusion is smaller than or equal to the premise. 
Furthermore, the instances of \eqref{eq:ac1} and \eqref{eq:ac3} used to rewrite $t_1(\dotsb t_n)$ into $t'_1(\dotsb t'_n)$ are always proper instances; 
\AD{This whole part probably needs to be clarified. Not sure how best to write it.}
to see why, consider the cases where we would need to rewrite an occurrence of $xy$ or $x(yz)$ (all distinct variables). 
2 subterms: 
\AD{Define what this means? Means $\#\acsubterms_f(\text{term}) = 2$}
$xy$, and there is nothing to rewrite since $xy \not\succ_\text{lex} yx$. 
3 subterms: 
if the term is $s(yz)$, with $y,z$ variables and $s$ any term, there is also nothing to rewrite for the same reason. 
4 subterms: 
it would be admissible to rewrite 
$t(s(xy))$ to $s(t(yx))$ 
if $s \prec t$, 
but this can be done without using $xy\bto yx$ directly, 
by ``moving'' $t$ to the right, using proper instances of $AC_f$ to swap $x$ and $y$, then ``moving'' $t$ back to the desired place (e.g.\ $t(s(xy)) \to s(t(xy)) \to s(x(ty)) \to s(x(yt)) \to s(y(xt)) \to s(y(tx)) \to s(t(yx))$). 
5 subterms: 
the term has form $t(s(x(yz)))$, 
so an identical process as for the 4 subterm case applies. 
6 subterms and more:
the term has form $u(t(s(x(yz))))$, 
so the case for 5 subterms appears at a subterm position. 
\qed\end{proof}

In practice, this criterion can be implemented by applying the following function 
\begin{alignat}{2}
&\acnormalise_f(s_1(\dotsb s_n)) & {}&={} \tabularbox{%
	let $\acsort_f(\acnormalise_f(s_1),\dotsc,\acnormalise_f(s_n)) = (s'_1,\dotsc,s'_n)$ \\
	in $s'_1(\dotsm s'_n)$
} \nonumber \\
&\acnormalise_f(g(t_1,\dotsc,t_n)) & {}&={} g(\acnormalise_f(t_1),\dotsc,\acnormalise_f(t_n)) \qc{\text{if $g \ne f$}}
\end{alignat}
to all literals in the clause, where \AD{Better notation here}
\begin{align}
\acsort_f(s_1,\dotsc,s_n) = \begin{dcases}
	s_k \concat \acsort^*_f (s_1,\dotsc,s_n \setminus s_k) 
		&\tabularbox{%
			if $\Exists s_k \in \{s_1,\dotsc,s_n\} \such~ s_k \prec_{t} s_1$ \\
			and $s_k$ minimal in $\{s_1,\dotsc,s_n\}$
		} \\
	s_1 \concat \acsort_f (s_2,\dotsc,s_n) 
		&\text{otherwise}
\end{dcases}
\end{align}
and $\acsort^*_f$ orders the list of terms using some total extension of the term ordering. 

Some examples, assume $g(\ldots) \succ b \succ a$: 
\KK{[FIXED] define ordering $g(x) \succ b\succ a$} 
\KK{[FIXED] why the last example is valid ?}
\begin{subequations}
\begin{align}
b(xa) &\to a(xb) \\
x(ba) &\to x(ab) \\
g(x) \, (a x) &\to a (x \, g(x)) \\
g(bx) \, g(ba) &\to g(ab) \, g(bx)
\end{align}
\end{subequations}
note the rhs may not be unique 
(e.g.\ in the first and third), 
since we are free to extend the term ordering in any (consistent) way. 

The main advantages of applying this simplification rule are 
\begin{itemize}[wide]
\item 
Strictly more redundant clauses found. 
For example, in the set 
$\{ a(bx) , a(xb) ,\allowbreak x(ab) , b(xa) , b(ax) , x(ba) \}$, 
the latter three are redundant, instead of only the latter one. 

\item 
Faster implementation. 
Even for simplifications that were already allowed by demodulation, 
we avoid the work of searching in indices and instantiating the axioms to perform the rewrites. 
%
Also, we can avoid storing $AC_f$ in the demodulation indices entirely. 
Since \eqref{eq:ac1} matches with all $f$-terms, and \eqref{eq:ac3} with all $f$-terms with 3 or more elements, 
this makes all queries on those indices faster. 
\end{itemize}


\section{Experimental results}\label{sec:experiment}
We implemented the simplifications developed in this paper —
encompassment demodulation, AC joinability and AC normalisation — in a theorem prover for first-order logic, iProver~\cite{iprover-systemdesc,cade2020}.\footnote{iProver is available at \url{http://www.cs.man.ac.uk/~korovink/iprover}} 
iProver combines superposition  with Inst-Gen and resolution calculi. For superposition iProver implements a range of simplifications including demodulation, light normalisation, subsumption and subsumption resolution. 
We run our experiments over FOF problems of the TPTP v7.4 library~\cite{tptp} (\num{17053} problems) on a cluster of Linux servers with \SI{3}{GHz} 11 cores AMD CPUs, \SI{128}{GB} memory, each problem was running on a single core with time limit \SI{300}{s}.


In total iProver solved \num{10358} problems.  
Encompassment demodulation (excluding cases when usual demodulation is applicable) was used in \num{7283} problems, $\geq$ \num{1000} times in \num{2343} problems, $\geq$ \num{10000} in 1018 problems, and $\geq$ \num{100 000} in 272 problems. 
This is in addition to other places where usual demodulation is valid but an expensive ordering check is skipped. 

There are \num{1366} problems containing 1 to 6 AC symbols, as detected by iProver.
AC normalisation was applied in \num{1327} of these: $\geq \num{1000}$ times in \num{1047} problems, $\geq \num{10000}$ times in 757 problems; and  $\geq \num{100000}$ times in \num{565} problems.
AC joinability was applied in \num{1138} problems: $\geq \num{1000}$ times in \num{646}, $\geq \num{10000}$ times in \num{255} problems. 
%
We can conclude that new simplifications described in this paper were applicable in a large number of problems 
and were used many times. 

\todo{Generally, it would be good to have more experimental results: number of newly solved problems}

\section{Conclusion and future work}\label{sec:discussion}
In this paper we extended the AC joinability criterion 
 to the superposition calculus for full first-order logic. 
For this we introduced a new closure-based redundancy criterion and proved that it preserves completeness. 
Using this criterion we proved that AC joinability and AC normalisation simplifications preserve completeness of the superposition calculus. 
Using these results, superposition provers for full first-order logic can incorporate AC simplifications without compromising completeness. 
Moreover, we extended demodulation to encompassment demodulation, which enables simplification of more clauses (and faster), independent of AC theories. 

We believe that the framework of closure redundancy can be used to prove many other interesting and useful redundancy criteria. 
For future work we are currently exploring other such applications, 
including more AC simplifications 
as well as general ground joinability criteria which can be incorporated in our framework. 



\bibliographystyle{plain}
\bibliography{bibliography}

\end{document}